\documentclass[a4paper,10pt]{article}
\usepackage[T1]{fontenc}
\usepackage[english]{babel}
\usepackage{amssymb,amsmath,latexsym,epsfig}

\usepackage[latin1]{inputenc}
\usepackage[T1]{fontenc}
\usepackage[all]{xy}
\usepackage{algpseudocode}

\makeatletter
\newcommand\ackname{Acknowledgements}
\if@titlepage
  \newenvironment{acknowledgements}{%
      \titlepage
      \null\vfil
      \@beginparpenalty\@lowpenalty
      \begin{center}%
        \bfseries \ackname
        \@endparpenalty\@M
      \end{center}}%
     {\par\vfil\null\endtitlepage}
\else
  
\fi
\makeatother

\title{Flags of almost affine codes}
\author{Trygve Johnsen\thanks{ Dept. of Mathematics, UiT The Arctic University of Norway, N-9037 Troms{\o}, Norway, \texttt{Trygve.Johnsen@uit.no}}  \and Hugues Verdure\thanks{ Dept. of Mathematics, UiT The Arctic University of Norway, N-9037 Troms{\o}, Norway, \texttt{Hugues.Verdure@uit.no}} }

\usepackage{bm}

\newtheorem{definition}{Definition}
\newtheorem{proposition}{Proposition}
\newtheorem{corollary}{Corollary}
\newtheorem{theorem}{Theorem}
\newtheorem{example}{Example}

\newtheorem{remark}{Remark}
\newtheorem{lemma}{Lemma}
\newenvironment{proof}[1][Proof]{\begin{trivlist}
\item[\hskip \labelsep {\bfseries #1}]}{\end{trivlist}}

\newcommand{\N}{\mathbb{N}}

\newcommand{\F}{\mathbb{F}}
\renewcommand{\-}{\backslash}
\newcommand{\E}{\mathcal{E}}

\newcommand{\qed}{\nobreak \ifvmode \relax \else
      \ifdim\lastskip<1.5em \hskip-\lastskip
      \hskip1.5em plus0em minus0.5em \fi \nobreak
      \vrule height0.75em width0.5em depth0.25em\fi}

\begin{document}

\maketitle

\begin{abstract}
We describe a two-party wire-tap channel of type II in the framework of almost affine codes. Its cryptological performance is related to some relative profiles of a pair of almost affine codes. These profiles are analogues of relative generalized Hamming weights in the linear case.


\end{abstract}

%
%
%
%

\section{Introduction}

In~\cite{OW}, Ozarow and Wyner describe the wire-tap channel of type II. The idea is to encrypt a message $\bm{m}$ without the use of any key (public or private), and send the encrypted message on a noiseless channel. The intruder is able to listen to parts of the encrypted message. One is interested in designing a system that gives as little information as possible to the intruder. In their original paper, the authors give a system using linear codes: : given a $(n,n-k)$ parity check matrix $H$ of the linear code , the sender picks up uniformly and at random a preimage by $H$ of $\bm{m}$, where $\bm{m}$ is  seen as a column vector of length $n-k$, and sends this preimage of length $n$ over the channel.

In~\cite{wei91}, Wei relates the maximum amount of information collected by the intruder to the generalized Hamming weights of the dual of the linear code.

In~\cite{LMVC}, the authors look at a variant of Ozarow and Wyner's scheme, namely the so-called two-party wire-tap channel of type II. This time, the intruder, not only is able to listen to parts of the encrypted message, but also knows a part of the original message. This leads to the study of pairs of linear codes. In~\cite{LCL} and~\cite{ZBLH}, the authors looks look at relative generalized Hamming weights and relative dimension/length profiles of such pairs of codes, and they relate these quantities to the amount of information gathered by the intruder in the two-party wire-tap channel of type II.

Almost affine codes are a strict generalization of linear/affine codes, introduced by Simonis and Ashikhmin i~\cite{SA}. In~\cite{JV}, we give a scheme for the wire-tap channel of type II that uses almost affine codes. We also give an analogue of Wei's result relating the amount of information gained by the intruder to the generalized Hamming weights of the dual of the matroid associated to the almost affine code.

In the present paper, we look at the two-party wire-tap channel of type II for almost affine codes. We build on the scheme presented in~\cite{JV}. As in~\cite{LMVC}, the intruder is now allowed to listen to parts of the encrypted message and knows parts of the original message. We build a larger almost affine code that only depends on the part of the original message known to the intruder. At this point, we deviate from the scheme in ~\cite{LMVC}, where the subset of the original sympols tapped by the intruder, as described in that paper, corresponds to a (linear) subcode, say $C_2$, of $C$. 
Hence our scheme, specialized to the subclass of linear codes, is not exactly the same as that of \cite{LMVC}. If, however, in ~\cite{LMVC}, one simply chooses to look at the dualized codes,  one gets a set-up, which is closer to ours, since the tapped symbols of the original message then correspond to a larger code $C_2^*$ containing 
$C^*.$ Thus our scheme is an analogue of the one in \cite{LMVC}.

A unifying theme for all four set-ups (the ones in \cite{wei91}, \cite{LMVC}, \cite{JV}, and the present one) is that of rank functions.  Each linear or almost affine code defines a rank function of a matroid. Moreover each pair of a code and a subcode comes equipped by a rank function $\rho$, which is 
the difference of the two (matroid) rank function of the codes in question. The function $\rho$ is not
necessarily a matroid rank function, but it is a so-called demi-matroid rank function. Demi-matroids were defined and described in \cite{BJMS},  and so-called profiles of demi-matroids were further described in \cite{BJM}. The relative Hamming weights of pairs of linear codes, as described in \cite{LMVC}, are examples of such profiles.
We analyze the performance of our cryptological set-up, and find that it is described in detail by 
an essential profile of the associated demi-matroid. This is formulated in our main result, Theorem \ref{summingup}. 

The paper is organized as follows: in Section~\ref{prel}, we recall basic definitions and properties about matroids, demi-matroids and almost affine codes.

In Section~\ref{flags}, we study flags, or chains, of almost affine codes, and the demi-matroids formed  by the alternating sums of their rank functions. A main result is Proposition \ref{mainprop}, which is a generalization of \cite[Theorem 9]{BJM}. Flags of length two will be of particular interest in remaining part of the paper.

In Section~\ref{meat} we investigate a two-party wire tap channel of type II, and describe in detail how we use an almost affine code $C$ and a certain function $\varphi: A \times A \longrightarrow A$ for the code alphabet $A$, to encode messages. Moreover we describe another bigger almost affine code $D$ that corresponds to uncoded symbols tapped by an adversary. By means of an entropy analysis we arrive at our conclusions. We also give an example with a non-linear almost affine code, whose associated matroid is the non-Pappus matroid of length 9 and rank $3$.

%
%
%
%

\section{Matroids, demi-matroids and almost affine codes} \label{prel}

In this section, we essentially recall relevant material that will be needed in the sequel, and we do not claim to have any new results here. We refer to~\cite{O} for the theory of matroids, to~\cite{BJMS} for an introduction on demi-matroids and to~\cite{SA} for an introduction on almost affine codes, and we will use their notation.

\subsection{Matroids and demi-matroids}

A matroid is a combinatorial structure that extend the notion of linear (in)dependency. There are many equivalent definitions, but we will give just one here. 

\begin{definition}\label{def:matroid} A matroid is a pair $M=(E,r)$ where $E$ is a finite set, and $r$ a function on the power set of $E$ into $\N$ satisfying the following axioms: \begin{description}
\item[(R1)] $r(\emptyset)=0$,
\item[(R2)] for every subset  $X \subset E$ and $x \in E$, $r(X) \leqslant r(X \cup \{x\}) \leqslant r(X)+1,$
\item[(R3)] for every $ X \subset E$ and $x,y \in E$, if $r(X) = r(X \cup \{x\}) = r(X \cup \{y\})$, then $r(X \cup \{x,y\})=r(X).$
\end{description}
\end{definition}

The set $E$ is called the ground set, and the function $r$ the rank function. Any subset $X \subset E$ with $r(X)=|X|$ is called independent. A basis is a inclusion maximal independent set. Finally, the rank $r(M)$ of the matroid $M$ is $r(E)$, and is also the cardinality of any basis.

Demi-matroids were introduced in~\cite{BJMS} and elaborated on in~\cite{BJM} when the authors studied flags of linear codes. They are a generalization of matroids in the following way:

\begin{definition}\label{def:demi} A demi-matroid is a pair $M=(E,r)$ where $E$ is a finite set, and $r$ a function on the power set of $E$ into $\N$ satisfying axioms (R1) and (R2) above.
\end{definition}

\begin{remark} The definitions in~\cite{BJMS,BJM} are different, but equivalent. See for example~\cite[Theorem 5.5]{gordon12} 
\end{remark}

As for matroids, the set $E$ is called the ground set, and the function $r$ the rank function. The rank $r(M)$ of the demi-matroid $M$ is $r(E)$. \\

Matroids and demi-matroids have duals defined in the following way: 

\begin{proposition}\label{prop:firstdual}Let $M=(E,r)$ be a matroid (respectively a demi-matroid). Then $M^*=(E,r^*)$ with $r^*$ defined as \[r^*(X) =|X|+r(E \- X)-r(E)\] is a matroid (respectively a demi-matroid). Moreover, $\left(M^*\right)^*=M$.
\end{proposition}

\begin{proof} The fact that $r^*$ is the rank function of a matroid if $r$ is so, is contained in any standard textbook on matroids, e.g. \cite{O}. The fact that $r^*$ is the rank function of a demi-matroid if $r$ is so, is immediate from axioms (R1) and (R2).
\qed\end{proof}

The matroid (respectively demi-matroid) $M^*$ is called the dual (respectively the dual or first dual) of $M$. It has rank $|E|-r(M)$. Demi-matroids have another dual, called the supplement dual or second dual:

\begin{proposition}\label{prop:seconddual}Let $M=(E,r)$ be a demi-matroid. Then $\overline{M}=(E,\overline{r})$ with $\overline{r}$ defined as \[\overline{r}(X)=r(E)-r(E\-X)\] is a demi-matroid. Moreover, we have $\overline{\overline{M}}=M$ and $\overline{M^*}=\overline{M}^*$. 
\end{proposition}

\begin{proof}See~\cite[Theorem 4]{BJMS}.
\qed\end{proof}

\subsection{Almost affine codes}
Almost affine codes were first introduced in~\cite{SA}, and are a combinatorial generalization of affine codes. 

\begin{definition}An almost affine code on a finite alphabet $A$, of length $n$ and dimension $k$ is a subset $C \subset A^n$ such that $|C|=|A|^k$ and such that for every subset $X \subset E=\{1,\cdots,n\}$, \[\log_{|A|}|C_X| \in \N,\] where $C_X$ is the puncturing of $C$ with respect to $E\-X$.

An almost affine subcode of $C$ is a subset $D \subset C$ which is itself an almost affine code on the same alphabet.
\end{definition}

\begin{remark} Any linear or affine code is obviously an almost affine code.
\end{remark}

To any almost affine code $C$ of length $n$ and dimension $k$ on the alphabet $A$, we can associate a matroid $M_C$ on the ground set $E=\{1,\cdots,n\}$ and with rank function \[r(X) = \log_{|A|} |C_X|,\] for $X \subset E$.

\begin{example}\label{running}Consider the almost affine code $C'$ in~\cite[Example 5]{SA}. It is a code of length $3$ and dimension $2$ on the alphabet $A=\{0,1,2,3\}$. Its set of codewords is \begin{align*} 000 && 011 && 022 && 033\\101&&112&&123&&130 \\ 202&&213&&220&&231\\303&&310&&321&&332\end{align*} Its matroid is the uniform matroid $U_{2,3}$ of rank $2$ on $3$ elements. This is an example of an almost affine code which is not equivalent to a linear code. \end{example}

While it is easy to see that any linear code has subcodes of a given dimension (and we even know the number of such subcodes), it is not straightforward to do so for almost affine codes. This is what we summarize here now.

\begin{definition} Let $C$ be a block code of length $n$, and let $\bm{c} \in C$ be fixed. The $\bm{c}$-support of any codeword $\bm{w}$ is \[Supp(\bm{w},{\bm{c}}) = \{i, \bm{c}_i \neq {\bm{w}}_i\}.\] The $\bm{c}$-support of $C$ is \[Supp(C,\bm{c}) =\bigcup_{\bm{w} \in C} Supp(\bm{w},{\bm{c}}).\]
\end{definition}

Note that the $\bm{c}$-support of the code is independent of the choice of $\bm{c} \in C$ (see~\cite[Lemma 1]{JV}), and it will therefore be denoted by $Supp(C)$ without reference to any codeword.

\begin{definition} \label{fixed}
Let $C$ be an almost affine code of length $n$, and let ${\bm{c}} \in F^n$ be fixed. Then \[C(X,{\bm{c}}) = \{\bm{w} \in C, \bm{w}_X = {\bm{c}}_X\},\] where $\bm{w}_X$ is the projection of $\bm{w}$ to $X$.
\end{definition} 

This might be empty, or not be an almost affine code, but when we take $\bm{c} \in C$, we get the following (\cite[Corollary 1]{SA}): 

\begin{proposition}\label{prop2} Let $C$ be an almost affine code of length $n$ and dimension $k$ on the alphabet $A$. Let ${\bm{c}} \in C$. Let $X \subset \{1,\cdots,n\}$. Then $C(X,{\bm{c}})$ is an almost affine subcode of $C$, and moreover, the rank function $\rho$ of the matroid associated to $C(X,\bm{c})$ is given by \[\rho(Y)= r(X \cup Y)-r(X)\] where $r$ is the rank function of the matroid $M_C$. In particular, \[|C(X,{\bm{c}})| = |A|^{k-r(X)}.\] 
\end{proposition}

\begin{corollary}\label{subcodes} Every almost affine code $C$ of dimension $k$ has almost affine subcodes of dimension $0\leqslant i \leqslant k$.
\end{corollary}

%
%
%
%

\section{Flags of demi-matroids and almost affine codes} \label{flags}

In this section, we look at flags of (demi-)matroids and almost affine codes. We show that under certain assumptions, they give rise to a demi-matroid in a natural way. These assumptions are always satisfied for flags of almost affine codes. We are mainly interested in pairs of codes in the next section, but for the sake of generality, we look at flags.

\subsection{Flags of demi-matroids}

\begin{definition} A flag of demi-matroids is a finite set of demi-matroids $M_i=(E_i,r_i)$ for $1 \leqslant i \leqslant m$ on the same ground set $E$, and such that \[\forall X \subset E,\ r_m(X) \leqslant \cdots \leqslant r_2(X) \leqslant r_1(X).\] A pair of demi-matroids is a flag with two demi-matroids.
\end{definition}

Given a flag of demi-matroids as above, we can define a function $\rho$ on the power set of $E$ as the alternative sum of the rank functions, namely \[ \forall X \subset E, \rho(X)=\sum_{i=1}^m (-1)^{i+1}r_i(X).\] It is natural to ask when the pair $M=(E,\rho)$ is a demi-matroid.

\begin{definition} Let $M=(E,r)$ be a demi-matroid. Let \[\mathcal{E}_M = \left\{(X,x),\ X \subset E,\ x \in X,\ r(X\-\{x\})=r(X)\right\}.\]
\end{definition}

\begin{theorem}\label{pairdm}Let $(M_1, M_2)$ be a pair of demi-matroids on the ground set $E$. Then $(E,\rho)$, where $\rho$ is defined above, is a demi-matroid if and only if $\E_{M_1} \subset \E_{M_2}$.
\end{theorem}

\begin{proof}Suppose first that $(E,\rho)$ is a demi-matroid. Then, for every $X \subset E$ and any $x \in X$, we have \[\rho(X\-\{x\}) \leqslant \rho(X).\] Assume that $(X,x) \in \E_{M_1}$. Then \begin{eqnarray*} r_2(X)-r_2(X\-\{x\}) &=& r_1(X)-\rho(X) -r_1(X\-\{x\})+\rho(X\-\{x\}) \\&=& r_1(X)-r_1(X\-\{x\}) - (\rho(X)-\rho(X\-\{x\})) \\&=& 0 - (\rho(X)-\rho(X\-\{x\})) \\&\leqslant& 0.\end{eqnarray*} But we also have \[r_2(X)-r_2(X\-\{x\}) \geqslant 0\] which proves that this is in $\E_{M_2}$ too.

Assume now that $\E_{M_1} \subset \E_{M_2}$. (R1) is trivially fulfilled by $\rho$. And let $X \subset E$, $x \in E$. We compute \begin{eqnarray*} \rho(X \cup \{x\}) - \rho(X) &=& r_1(X \cup\{x\})-r_2(X \cup \{x\}) - (r_1(X)-r_2(X)) \\&=& (r_1(X \cup \{x\})-r_1(X))-(r_2(X \cup \{x\}) - r_2(X)) \end{eqnarray*} We have 3 cases \begin{itemize} \item If $(X \cup \{x\},x) \in \E_{M_1} \subset \E_{M_2}$, then \[\rho(X \cup \{x\}) - \rho(X) = (r_1(X \cup \{x\} - r_1(X)) - (r_2(X \cup \{x\} - r_2(X))= 0-0=0,\]
\item If $(X \cup \{x\},x) \in \E_{M_2} \- \E_{M_1}$, then \[\rho(X \cup \{x\}) - \rho(X) = (r_1(X \cup \{x\} - r_1(X)) - (r_2(X \cup \{x\} - r_2(X))= 1-0=1,\]
\item If $(X \cup \{x\},x) \not \in \E_{M_2}$, then \[\rho(X \cup \{x\}) - \rho(X) = (r_1(X \cup \{x\} - r_1(X)) - (r_2(X \cup \{x\} - r_2(X))= 1-1=0.\] \end{itemize} In any cases, \[\rho(X) \leqslant \rho(X \cup\{x\}) \leqslant \rho(X)+1.\]
\qed\end{proof}

We can now generalize to flags of demi-matroids. Note that this generalization is just an implication, not an equivalence.

\begin{theorem}\label{flagdm} Let $(M_1,\cdots,M_m)$ be a flag of matroids. If $\E_{M_1} \subset \E_{M_2} \subset \cdots \subset \E_{M_m}$, then $(E,\rho)$ with $\rho$ defined above is a demi-matroid.
\end{theorem}

\begin{proof} For simplicity, we write $\E_j$ for $\E_{M_j}$. The fact that $\rho(\emptyset)=0$ is trivial. Now, let $X \subset E$ and $x \in E$. Let $j$ be minimal such that $(X \cup \{x\},x) \in \E_j \- \E_{j-1}$ with the convention that $j=1$ if $(X \cup \{x\},x) \in \E_1$ and $j=m+1$ if $(X \cup \{x\},x) \not \in \E_m$.Then we have \begin{eqnarray*} \rho(X \cup \{x\}) &=& \sum_{i=1}^m (-1)^{i+1} r_i(X \cup \{x\}) \\&=& \sum_{i=1}^{j-1} (-1)^{i+1} r_i(X \cup \{x\})+\sum_{i=j}^m (-1)^{i+1} r_i(X \cup \{x\}) \\ 
 \\&=& \sum_{i=1}^{j-1} (-1)^{i+1} [r_i(X) +1]+\sum_{i=j}^m (-1)^{i+1} r_i(X ) \\ &=& \rho(X) + \sum_{i=1}^{j-1}(-1)^{i+1}
\end{eqnarray*} Independently on the parity of $j$, we always have \[0 \leqslant   \rho(X \cup \{x\}) -\rho(X) \leqslant 1.\]
\qed\end{proof}

\subsection{Flags of almost affine codes}

\begin{definition} A flag $F=(C_1,\cdots,C_m)$ of almost affine codes is a finite set of almost affine codes on the same alphabet and same length, with the property that for $1\leqslant j \leqslant m-1$, $C_{j+1}$ is an almost affine  subcode of $C_j$. A pair of almost affine codes is a flag with two codes.
\end{definition}

We will see in this section that flags of almost affine codes give rise to demi-matroids in a natural way. We will also look at the duals of these demi-matroids. We start with a lemma.

\begin{lemma}\label{lemma1} Let $C$ be an almost affine code, $\bm{w} \in C$. Let $X \subset E$ and $x \in E\-X$. Then \[r(X \cup\{x\}) = r(X) \Leftrightarrow C(X \cup\{x\},\bm{w}) = C(X,\bm{w}).\]
\end{lemma}

\begin{proof} We always have $C(X \cup\{x\},\bm{w}) \subset C(X,\bm{w}).$ Then it is a direct consequence of Proposition~\ref{prop2}. 
\qed\end{proof}

\begin{lemma}\label{lemma2} Let $(C_1,C_2)$ be a pair of almost affine codes. Let $\bm{w} \in C_2$, and $X \subset E$. Then \[C_2(X,\bm{w}) = C_2 \cap C_1(X,\bm{w}).\]
\end{lemma}

\begin{proof} This is obvious.
\qed\end{proof}

Given a flag of almost affine codes $F=(C_1, \cdots ,C_m)$ of length $n$, define the following function on the power set of $E=\{1,\cdots,n\}$ by \[\rho_F(X) = \sum_{i=1}^m (-1)^{i+1} r_i(X)\] for $X \subset E$, where $r_i$ is the rank function of the matroid associated to $C_i$.

\begin{theorem} Let $F=(C_1 , \cdots , C_m)$ be a flag of almost affine codes. Then the pair $(E,\rho_F)$ defined above is a demi-matroid.
\end{theorem}

\begin{proof} As before, we write $\E_j$ instead of $\E_{M_j}$, where $M_j$ is the matroid associated to $C_j$. By Theorem~\ref{flagdm}, it suffices to show that $\E_{j} \subset \E_{j+1}$, for $1\leqslant j \leqslant m-1$. Let $(X,x) \in \E_j$ and $\bm{w} \in C_{j+1}\subset C_j$. From Lemma~\ref{lemma1}, this means that \[C_j(X\-\{x\},\bm{w}) = C_j(X,\bm{w}).\] 
Then by Lemma~\ref{lemma2}, \[C_{j+1}(X\-\{x\},\bm{w}) = C_{j+1}(X,\bm{w}),\] and by Lemma~\ref{lemma1} again, $(X,x) \in \E_{j+1}.$
\qed\end{proof}

Almost affine codes generally do not have duals. Their associated (demi-)matroids have, and we could have defined functions using $r_i^*$, $\overline{r_i}$ or $\overline{r_i^*}$. In the sequel we see what are the relations between these alternative functions.

\begin{lemma}
Let $M_1=(E,r_1)$ and $M_2=(E,r_2)$ be two demi-matroids. Then \[\E_{M_1} \subset \E_{M_2} \Leftrightarrow \E_{M_2^*} \subset \E_{M_1^*} \Leftrightarrow \E_{\overline{M_1}} \subset \E_{\overline{M_2}}.\] 
\end{lemma} 

\begin{proof} Suppose that $\E_{M_1} \subset \E_{M_2}$. Let $(X,x) \in \E_{M_2^*}$. Then \[|X\-\{x\}| + r_2(E\-(X\-\{x\})) - r_2(E) = |X|+r_2(E\-X)-r_2(E),\] that is \[r_2(E\-X) = r_2((E\-X) \cup \{x\}) -1.\] In other words, $((E\-X)\cup\{x\}) \not \in \E_{M_2}$ and consequently $((E\-X)\cup\{x\}) \not \in \E_{M_1}$ either, that is \[r_1(E\-X) = r_1((E\-X) \cup \{x\}) -1\] by the definition of $\E_{M_1}$ and (R2).    Then \begin{eqnarray*} r_1^*(X\-\{x\}) &=& |X\-\{x\}| + r_1(E\-(X\-\{x\})) - r_1(E) \\&=& |X|-1 + r_1(E\-X)+1 -r_1(E) \\&=& r_1^*(X),
\end{eqnarray*} that is $(X,x) \in \E_{M_1^*}.$
Let now $(X,x) \in \E_{\overline{M_1}}$, that is \[r_1(E)-r_1(E\-(X\-\{x\})) = \overline{r_1}(X\-\{x\}) = \overline{r_1}(X) = r_1(E)-r_1(E\-X).\] Then $((E\-X) \cup\{x\},x) \in \E_{M_1} \subset \E_{M_2}$, and by the same computation for $\overline{r_2}$, $(X,x) \in \E_{\overline{M_2}}$.\\
The two other implications follow by duality.
\qed\end{proof}

\begin{proposition} \label{mainprop}
Let $F=(C_1,\cdots,C_m)$ be a flag of almost affine codes, with $r_i$ be the rank function of the matroid associated to the code $C_i$ for $1\leqslant i\leqslant m$. Define the following functions:
\begin{eqnarray*} \eta_F = \sum_{i=1}^m(-1)^{m-i}r_i^* , \\
\theta_F= \sum_{i=1}^m(-1)^{i+1}\overline{r_i},\\
\pi_F= \sum_{i=1}^m(-1)^{m-i}\overline{r_i^*}.
\end{eqnarray*} Then $(E,\eta_F)$, $(E,\theta_F)$ and $(E,\pi_F)$ are all demi-matroids. Moreover, we  have the following duality relations: \[\begin{array}{ccc} \eta_F=\left\{\begin{array}{ll}\overline{\rho_F} & \textrm{ if }m\textrm{ is even} \\ \rho_F^*& \textrm{ if }m\textrm{ is odd}  \end{array}\right., & \theta_F=\overline{\rho_F},&   \pi_F=\left\{\begin{array}{ll}{\rho_F} & \textrm{ if }m\textrm{ is even} \\ \overline{\rho_F}^*& \textrm{ if }m\textrm{ is odd}  \end{array}\right..\end{array}\]
\end{proposition}

\begin{proof}
The first part of the proposition is an easy adaptation of the proof of the previous theorem, together with the previous lemma. For the second part, we compute \begin{eqnarray*} \eta_F(X) &=& \sum_{i=1}^m(-1)^{m-i} r_i^*(X) \\ &=& \sum_{i=1}^m(-1)^{m-i}\left[ |X| + r_i(E\-X)-r_i(E)\right] \\ &=& \sum_{i=1}^m(-1)^{m-i}|X| + \sum_{i=1}^m(-1)^{m-i}r_i(E\-X) -\sum_{i=1}^m(-1)^{m-i}r_i(E) \\&=& \left\{\begin{array}{ll} 0-\sum_{i=1}^{m}(-1)^{i+1}r_i(E\-X) + \sum_{i=1}^m(-1)^{i+1} r_i(E) & \textrm{ if }m\textrm{ is even} \\ 
|X| + \sum_{i=1}^{m}(-1)^{i+1}r_i(E\-X) - \sum_{i=1}^m(-1)^{i+1} r_i(E)& \textrm{ if }m\textrm{ is odd }  
\end{array}\right.
\\&=&\left\{\begin{array}{ll}\overline{\rho_F} & \textrm{ if }m\textrm{ is even} \\ \rho_F^*& \textrm{ if }m\textrm{ is odd}  \end{array}\right. 
\end{eqnarray*} The other equalities are done in a similar way.
\qed\end{proof}

\begin{remark} \label{expl} The corollary generalizes~\cite[Theorem 9]{BJM}, where the corresponding result was proven for flags of linear codes. In Part (b)
the $r_i^*$ are rank functions of the dual codes (orthogonal complements) $C_i^\perp$ if the $C_i$ are linear (or even multilinear) codes. If we only know that the $C_i$ are almost affine, we do not necessarily have dual codes, for which the $r_i^*$ are rank functions. See \cite{JV}. In Section \ref{meat}, where we give the results and applications that we hope are the most interesting ones, it is only the case $m=2$ that is considered. The case of longer flags ($m\geqslant 3$) was included above for completeness, and for possible usage in new, e.g. cryptological, applications that might turn up in the future (although we unfortunately have not found applications for chains of length $3$ or more yet).
\end{remark}

%
%
%
%

\section{Two-party wire tap channel of type II} \label{meat}

In~\cite{OW}, Ozarow and Wyner introduce the wire-tap channel of type II. The idea is to encode a message without the use of any key and send it over a channel where an intruder can listen to a subset of the transmitted symbols. The goal of the encoder is to minimize  the information about the original message the intruder can get. The authors propose a scheme using linear codes. In~\cite{wei91}, the amount of information collected by the intruder is related to the generalized Hamming weights of the dual of the linear code. In~\cite{JV}, we give a scheme that uses almost affine codes, and give an analogue of Wei's result.\\
In~\cite{LMVC}, the authors look at the following two-party wiretap channel of type II: this time, the intruder is able to get some symbols of the original message, as well as to listen to a subset of the transmitted symbols. In this section, we investigate the same scenario for our scheme.\\

\subsection{Wiretap channel of type II for almost affine codes}

\noindent We start by recalling our scheme described in~\cite{JV}: let $C$ be an almost affine code of length $n$ and rank $k$ on the alphabet $A$. Let $B \subset E=\{1,\cdots,n\}$ be a basis of the matroid $M_C$. Let $\varphi: A \times A \longrightarrow A$ be a function such that for every $y \in A$, $\varphi_{y,1}=\varphi(y, -)$ and $\varphi_{y,2}=\varphi(-,y)$ are bijections. For $\bm{m} \in A^{E\-B}$, let $\Phi_{\bm{m}} : A^E \longrightarrow A^E$ be defined the following way: if $\bm{w} \in A^E$, then \[\Phi_{\bm{m}}(\bm{w})_i = \left\{ \begin{array}{ll} \bm{w}_ i & \textrm{if } i \in B \\ \varphi(\bm{w}_i,\bm{m}_i) & \textrm{ if } i \in E\-B\end{array}\right.\] We can then define, for $\bm{m} \in A^{E\-B}$, \[C_{\varphi,\bm{m}} = \Phi_{\bm{m}}(C).\] Whenever $\varphi$ is obvious from the context, we may omit it and write $C_{\bm{m}}$ for $C_{\varphi,\bm{m}}$.
From~\cite{JV}, we have the following: \begin{lemma}\label{Cm} For every $\bm{m} \in A^{E\-B}$, $C_{\bm{m}}$ is an almost affine code of length $n$ and dimension $k$ over $A$, with associated matroid equal to $M_C$. Moreover, the codes $C_{\bm{m}}$ for $\bm{m} \in A^{E\-B}$ form a partition of $A^E$.
\end{lemma}

The scheme is the following: one wishes to send the message $\bm{m}$ over a channel. The sender chooses uniformly and at random an element  $\bm{w} \in C_{\bm{m}}$. The receiver finds the unique $\bm{m'}$ such that $\bm{w} \in C_{\bm{m'}}$. Then $\bm{m'}=\bm{m}$.

In~\cite{JV}, an intruder was able to listen to up to $\mu$ of the $n$ symbols sent over the channel. Then it is shown that whenever the intruder is able to listen up to $d^*_{j+1}$ symbols of the transmitted message, then he has knowledge of at most $j$ symbols of the original message, where $d_j^*$ are the generalized Hamming weights of $M_C^*$.

\subsection{Two-party wiretap channel of type II for almost affine codes}

In the two-party wiretap channel of type II, the intruder is able to get a subset $X \subset E\-B$ of the original message $\bm{m}$, and to tap a subset $Y \subset E$ of cardinality $\mu$ of the transmitted message over the channel. We will now look into the equivocation of the system, that is the maximum of information gained by the intruder after these taps.

\subsubsection{An almost affine overcode}

Let $\bm{M} \in A^X$. Define \[D_{X,\bm{M}} = \bigcup_{\begin{array}{c} \bm{m} \in A^{E\-B} \\ \bm{m}_X = \bm{M} \end{array}} C_{\bm{m}}.\]

\begin{lemma} For every $X \subset E\-B$ and every $\bm{M} \in A^X$, $D_{X,\bm{M}}$ is an almost affine code of length $n$ and rank $n-|X|$ on $A$. Its associated matroid $M_D$ has rank function \[r_{D_{X,\bm{M}}}(Y) =|Y\backslash(B\cup X))|+r(Y\cap (B\cup X)).\] It is independent of $\bm{M}$.
\end{lemma}

\begin{proof} For simplicity, we write $D$ for $D_{X,\bm{M}}$. Let $Y \subset E$. We have to show that $|D_Y|$ is a power of $|A|$.
We define the following sets:\begin{itemize} \item $Y_1 = Y\cap (X \cup B)$,
\item $Y_2= Y\-Y_1$,
\item $I \subset Y_1$ a maximal independent set of $Y_1$ for the matroid associated to $C$ (and therefore all $C_{\bm{m}}$). Since $B$ is a basis of the matroid, we may even assume that $Y_1 \cap B \subset I$.
\item $J=Y_1 \- I$,
\item $I_B=B \cap I=B \cap Y$,
\item $I_X=X \cap I=I\-I_B$.
\end{itemize} The idea is to show that on one hand we have full freedom on the choice of the characters on $Y_2$ and $I$ for the words of $D_Y$, the first one because of the choice of $\bm{m}$, the second one because $I$ is an independent set, which shows that we have  $|D_Y|\geqslant |A|^{|Y_2\cup I|}$. On the other hand, we show that if two words of $D_Y$ agree on $Y_2 \cup I$, then they must agree on the rest of $Y$ as well.

So let $\bm{a}=(a_i)_{i \in Y_2 \cup I} \in A^{Y_2 \cup I}$. Let $\bm{b}=(b_j)_{j \in I} \in A^I$ defined by \[b_j=\left\{ \begin{array}{ll}a_j& \textrm{ if } j \in I_B \\ 
\varphi^{-1}_{\bm{M}_j,2}(a_j) 
& \textrm{ if }j\in I_X\end{array}\right.\]

Since $I$ is an independent set for the matroid $M_C$, there exists a word $\bm{c}=(c_j)_{j \in E} \in C$ with \[\bm{c}_I = \bm{b}.\]Now, define $\bm{m}=(m_i)_{i \in E\-B}$ in the following way : \[m_i=\left\{ \begin{array}{ll}\bm{M}_i & \textrm{ if }i \in X \\ \varphi^{-1}_{c_i,1}(a_i) & \textrm{ if }i \in Y_2 \\\textrm{randomly} &   \textrm{otherwise} \end{array}\right.\]

Of course, $\bm{m}_X=\bm{M}$. Finally, look at the word $\bm{w}= \Phi_{\bm{m}}(\bm{c}).$ Then, by construction, $\bm{w} \in C_{\bm{m}} \subset D$ and $\bm{w}_{I \cup Y_2} = \bm{a}.$ 

Now, let $\bm{w_1},\bm{w_2} \in D$ such that $(\bm{w_1})_{I \cup Y_2}= (\bm{w_2})_{I \cup Y_2}$. By definition, there exists $\bm{c_1},\bm{c_2} \in C$ and $\bm{m_1},\bm{m_2} \in A^{E\-B}$ with $(\bm{m_1})_X = (\bm{m_2})_X = \bm{M}$ such that $\bm{w_1}=\Phi_{\bm{m_1}}(\bm{c_1})$ and $\bm{w_2}=\Phi_{\bm{m_2}}(\bm{c_2})$.

In particular, we have that \[(\bm{c_1})_{I_B} = (\Phi_{\bm{m_1}}(\bm{c_1}))_{I_B} = (\bm{w_1})_{I_B} = (\bm{w_2})_{I_B} = (\Phi_{\bm{m_2}}(\bm{c_2}))_{I_B} = (\bm{c_2})_{I_B}\] and for every $i \in I_X$, \[\varphi((\bm{c_1})_i,(\bm{m_1})_i) =(\Phi_{\bm{m_1}}(\bm{c_1}))_i= (\bm{w_1})_i = (\bm{w_2})_i =(\Phi_{\bm{m_2}}(\bm{c_2}))_i=  \varphi((\bm{c_2})_i,(\bm{m_2})_i).\] 
Now, since $(\bm{m_1})_i = \bm{M}_i = (\bm{m_2})_i$ and $\varphi_{\bm{M}_i,2}$ is a bijection, this shows that \[(\bm{c_1})_{I_X} = (\bm{c_2})_{I_X}.\] Thus, \[(\bm{c_1})_{I} = (\bm{c_2})_{I}.\] 
Now, by Proposition~\ref{prop2}, we get that \[(\bm{c_1})_{Y_1} = (\bm{c_2})_{Y_1}\] since \[\left|C_{Y_1}(I,(\bm{c_1})_{Y_1})\right| = |A|^{r(Y_1) - r(I)} = 1.\] This concludes the proof that $D$ is an almost affine code since: \begin{itemize} \item for $i \in Y \cap B \subset Y_1$, \[(\bm{w_1})_i = (\bm{c_1})_i = (\bm{c_2})_i = (\bm{w_2})_i,\]
\item for $i \in Y \cap X \subset Y_1,$ \[(\bm{w_1})_i = \varphi((\bm{c_1})_i,(\bm{m_1})_i)= \varphi((\bm{c_1})_i,\bm{M}_i)=\varphi((\bm{c_1})_i,(\bm{m_2})_i)= (\bm{w_2})_i,\]
\item for $i \in Y_2$, \[(\bm{w_1})_i =  (\bm{w_2})_i\] by assumption.
\end{itemize} But this also shows that the rank function of $M_D$ is given by \[ r_{D_{X,\bm{M}}}(Y) = |I \cup Y_2| = r(Y \cap (B \cup X)) + |Y \backslash (B \cup X)| ,\] which in turn shows that this just depends on $X$, not on $\bm{M} \in A^X$.
\qed\end{proof}

\begin{remark} For every $\bm{m} \in A^{E\-B}$ such that $(\bm{m})_X = \bm{M}$, the code $C_{\bm{m}}$ is an almost affine subcode of $D_{X,\bm{M}}$. They define a pair of codes, whose associated demi-matroid will be denoted $(E,\rho_X)$.
\end{remark}

\subsubsection{Conditional entropy of the system}

We want to see how much an intruder gets of information by listening to $\mu$ digits of the sent message, and knowing already a subset $X$ of the digits of the message $\bm{m}$. We will need two lemmas. The first one tells, given that the original message was $\bm{m} \in A^{E\-B}$ and that the intruder listens to $\bm{t} \in A^Y$,  how many possible $\bm{w} \in C_{\bm{m}}$ might have been sent over the channel. The second one tells how many messages $\bm{m} \in A^{E\-B}$ may have possibly be sent, given that $\bm{m}_X=\bm{M}$, and given that the intruder listens   to $\bm{t} \in A^Y$.

\begin{lemma}\label{omega} Let $\bm{M} \in A^X$ and $Y \subset E$. Let $\bm{m} \in A^{E\-B}$ and $\bm{t} \in A^Y$. Define \[\Omega_{\bm{t},Y,\bm{M}}(\bm{m})=\begin{cases} \emptyset & \text{ if } \bm{m}_X \neq \bm{M}, \\ \{\bm{w} \in C_{\bm{m}}, \ \bm{w}_Y = \bm{t}\} & \text{ otherwise}\end{cases}\] Then $\Omega_{\bm{t},Y,\bm{M}}(\bm{m})$ is either empty or has cardinality $|A|^{k-r(Y)}$.
\end{lemma}

\begin{proof} This is essentially the first part of~\cite[Lemma 9]{JV}.
\qed\end{proof}

\begin{lemma}\label{omeganonnul} Let $\bm{M} \in A^X$ and $\bm{w} \in D_{X,\bm{M}}$. Let $Y \subset E$ and $\bm{t}=\bm{w}_Y$. Then \[\left|\{\bm{m} \in A^{E\-B},\ \Omega_{\bm{t},Y,\bm{M}}(\bm{m}) \neq \emptyset\}\right| = |A|^{n-k-|X|-\rho_X(Y)}.\]
\end{lemma}

\begin{proof} We compute $|\{\bm{v} \in D_{X,\bm{M}},\ \bm{v}_Y=\bm{t}\}|$ in two different ways. Since $D_{X,\bm{M}}$ is a disjoint union of $C_{\bm{m}}$, we have \begin{eqnarray*} |\{\bm{v} \in D_{X,\bm{M}},\ \bm{v}_Y=\bm{t}\}| &=& \sum_{\bm{m}\in A^{E_B},\ \bm{m}_X=\bm{M}}|\{\bm{v} \in C_{\bm{m}},\ \bm{v}_Y=\bm{t}\}|\\&=& \sum_{\bm{m}\in A^{E\-B},\ \bm{m}_X=\bm{M}}|\Omega_{\bm{t},Y,\bm{M}}(\bm{m})| \\ &=&\sum_{\bm{m}\in A^{E\-B},\ \Omega_{\bm{t},x,\bm{M}}(\bm{m}) \neq \emptyset}|\Omega_{\bm{t},Y,\bm{M}}(\bm{m})|\\&=& \sum_{\bm{m}\in A^{E\-B},\ \Omega_{\bm{t},x,\bm{M}}(\bm{m}) \neq \emptyset}|A|^{k-r(Y)} \\&=& |\{\bm{m} \in A^{E\-B},\ \Omega_{\bm{t},Y,\bm{M}}(\bm{m}) \neq \emptyset\}||A|^{k-r(Y)}.\end{eqnarray*}

On the other hand, since $D_{X,\bm{M}}$ is an almost affine code, and $\bm{w} \in D_{X,\bm{M}}$, by~\cite[Proposition 2]{SA}, we have \[|\{\bm{v} \in D_{X,\bm{M}},\ \bm{v}_Y=\bm{t}\}| = |D_{X,\bm{M}}(\bm{w},Y)| = |A|^{r_{D_{X,\bm{M}}}(E)-r_{D_{X,\bm{M}}}(Y)}.\] The lemma follows easily.

\qed\end{proof}

A way of measuring how much an intruder gains information is the conditional entropy of the system, namely

\[H(\bm{m} | \bm{t},\bm{M})= \sum_{\bm{M} \in A^X, \bm{t} \in A^Y} p(\bm{M},\bm{t}) \sum_{\bm{m} \in A^{E\-B}} p(\bm{m}|\bm{t},\bm{M}) \log_{|A|}\frac{1}{p(\bm{m}|\bm{t},\bm{M})},\] where $p(\bm{M},\bm{t})$ is the probability to get both $\bm{M}$ and $\bm{t}$, while $p(\bm{m}|\bm{t},\bm{M})$ is the probability the probability that $\bm{m}$ is the original message, knowing that $\bm{m}_X=\bm{M}$ and that the intruder listens to  $\bm{t}$, and with the convention that $0 \log_{|A|} \frac{1}{0} = 0$. The conditional entropy is a measure of how many digits the intruder still does not know on the original message $\bm{m}$, if he has knowledge of a subset $X$ of the symbols of $\bm{m}$, and been able to listen to a subset $Y$ of the symbols of the message $\bm{w} \in C_{\bm{m}}$ sent through the channel. For example, if the intruder is able to listen to everything, then $p(\bm{m}|\bm{t},\bm{M})$ is either $0$ or $1$ (it is $1$ if and only if $\bm{m}$ is the original message), so the conditional entropy is $0$. On the other hand, if the intruder cannot tap any digit, then $p(\bm{M},\bm{t})=\frac{1}{|A|^{|X|}}$ and $p(\bm{m}|\bm{t},\bm{M})$ is either $0$ if $\bm{m}_x \neq \bm{M}$ or $\frac{1}{|A|^{|E \backslash B|-|X|}}$ if $\bm{m}_x = \bm{M}$, so that the conditional entropy is $n-k-|X|$, that is the number of digits of the original message (of length $n-k$) that are still unknown ($|X|$ digits are known already).

\begin{remark}\label{digit}When we say that the intruder knows $s$ extra digits, this means that he knows the equivalent of $s$ digits, not necessarily $s$ actual digits, but that $s$ degrees of freedom have been removed. 
\end{remark}

\begin{theorem} Suppose that $\bm{m} \in A^{E\-B}$ and $\bm{t} \in C_{\bm{m}}$ are chosen randomly and uniformly. Then the conditional entropy of the system is \[H(\bm{m} | \bm{t},\bm{M}) = n-k-|X|-\rho_X(Y).\] 
\end{theorem}

\begin{proof}By Lemma~\ref{omeganonnul}, for a given $\bm{M} \in A^X$ and $\bm{t} \in A^Y$, we have \[p(\bm{m}|\bm{t},\bm{M}) = \left\{\begin{array}{ll} 0 & \textrm{ if } \Omega_{\bm{t},Y,\bm{M}}(\bm{m}) = \emptyset \\ 
\frac{1}{|A|^{n-k-|X|-\rho_X(Y)}}& \textrm{ otherwise }
\end{array}\right..\] 

Thus \[\begin{split}\lefteqn{H(\bm{m} | \bm{t},\bm{M}) }\\&= \sum_{\bm{M} \in A^X, \bm{t} \in A^Y} p(\bm{M},\bm{t}) \sum_{\bm{m} \in A^{E\-B}} p(\bm{m}|\bm{t},\bm{M}) \log_{|A|}\frac{1}{p(\bm{m}|\bm{t},\bm{M})} \\&= \sum_{\substack{\bm{M} \in A^X, \bm{t} \in A^Y \\ p(\bm{M},t) \neq 0}} p(\bm{M},\bm{t}) \sum_{\substack{\bm{m} \in A^{E\-B} \\ \Omega_{\bm{t},Y,\bm{M}}(\bm{m}) \neq \emptyset}} p(\bm{m}|\bm{t},\bm{M}) \log_{|A|}\frac{1}{p(\bm{m}|\bm{t},\bm{M})} \\ &= \sum_{\substack{\bm{M} \in A^X, \bm{t} \in A^Y \\ p(\bm{M},t) \neq 0}} p(\bm{M},\bm{t}) \sum_{\substack{\bm{m} \in A^{E\-B} \\ \Omega_{\bm{t},Y,\bm{M}}(\bm{m}) \neq \emptyset}}\frac{n-k-|X|-\rho_X(Y)}
{|A|^{n-k-|X|-\rho_X(Y)}}
\\&=  \sum_{\substack{\bm{M} \in A^X, \bm{t} \in A^Y \\ p(\bm{M},t) \neq 0}} p(\bm{M},\bm{t})(n-k-|X|-\rho_X(Y))
\\&
=n-k-|X| - \rho_X(Y).
\end{split}\]
\qed\end{proof}

\subsubsection{Equivocation of the system}

We are interested in minimizing the amount of information an intruder may have access to  if he gets $X$ digits of the original message, and is able to listen to $\mu$ digits of the message sent over the channel, that is we are interested in maximizing the equivocation \[E_\mu = \min_{|Y|=\mu} H(\bm{m} | \bm{t},\bm{M}),\]  or equivalently minimizing \[\Delta_\mu = n-k-E_\mu,\] the maximum number of digits known to the intruder after listening to $\mu$ digits over the channel.

In~\cite{BJM}, the authors introduce the following profile of demi-matroids:

\begin{definition} \label{central}
Let $(E,r)$ be a demi-matroid of rank $l$, and let $0\leqslant i \leqslant l$. \[\sigma_i=\min\{|X|,\ r(X)=i\}.\]
\end{definition}

\begin{theorem} \label{summingup}
The uncertainty is given by  \[\Delta_\mu = |X|+j \Leftrightarrow \sigma_j \leqslant \mu < \sigma_{j+1}\] for $0 \leqslant j \leqslant n-|X|-k$ and with the convention that $\sigma_{n-|X|-k+1}= \infty$
\end{theorem}

\begin{proof}We have \[\Delta_\mu = n-k- \min_{|Y| =\mu} H(\bm{m} | \bm{t},\bm{M}) = |X| + \max_{|Y| = \mu} \rho_X(Y).\] 
It is easy to see that \[\max_ {|Y|=\mu} \{\rho_X(Y)\} =j \Rightarrow \mu \geqslant \sigma_j\] and that \[\mu \geqslant \sigma_j \Rightarrow \max_ {|Y|=\mu} \{\rho_X(Y)\}  \geqslant j.\] Then we get the equivalence \[\sigma_j \leqslant \mu < \sigma_{j+1} \Leftrightarrow \max_ {|Y|=\mu} \{\rho_X(Y)\} =j.\] The theorem follows.
\qed\end{proof}

\begin{remark} \label{turnaround}
{\rm In a remark on p. 1225 of \cite{LMVC} the authors note that in the analogous set-up there, using linear codes, the adversary can obtain at least $j$ information bits among-the non-tapped ones, if and only if he/she can tap at least $m_j$ of the encrypted (and transmitted) information bits. In that paper $m_j$ is the
(relative Hamming weight and) demi-matroid invariant 
\[m_j=\min\{|X|,\ \overline{R}(X)=j\}\]
for $R(Y)=r_1(Y)-r_2(Y),$ corresponding
to a pair $(C_1,C_2)$ appearing there.
But if instead one looks at 
$r(Y) = \overline{R}(Y)=r_2^*(Y)-r_1^*(Y)$,
corresponding to the pair of dual codes $(C_2^\perp,C_1^\perp)$, then the
$m_j$ described are precisely the \[\sigma_j=\min\{|X|,\ r(X)=j\}\]
described in Definition \ref{central} above. }
\end{remark}

\subsubsection{An example}

\setcounter{MaxMatrixCols}{18}
Let $A=\F_3^2$ and $\phi:\F_3^{18} \longrightarrow A^9$ defined by \[\phi(x_1,\cdots,x_{18}) = ((x_1,x_2),\cdots,(x_{17},x_{18})).\]  Consider the linear code $L$ over $\F_3$ with generator matrix \[G=\begin{bmatrix}
1&0&1&0&0&0&1&0&0&0&1&0&1&0&1&0&0&0\\
0&1&0&1&0&0&0&1&0&0&0&1&0&1&0&1&0&0\\
0&0&0&0&0&0&1&0&1&0&2&1&0&1&1&0&1&0\\
0&0&0&0&0&0&0&2&0&1&2&0&2&1&0&2&0&1\\
0&0&1&0&1&0&0&1&0&0&0&1&0&0&1&1&1&0\\
0&0&0&1&0&1&2&1&0&0&2&1&0&0&1&0&0&1\end{bmatrix}\]

Consider the folded code $C=\phi(L)$. It can be shown that this code is an almost affine code of length $9$ and dimension $3$ on the alphabet $A$. Moreover, its associated matroid is the non-Pappus matroid (see~\cite[Example 4]{SA}). As such, this code is not equivalent to any linear code. The independent sets of $M_C$ are all subsets of $E=\{1,\dots,9\}$ of cardinality at most $3$ except \[\{1,2,3\},\{1,5,7\},\{1,6,8\},\{2,4,7\},\{2,6,9\},\{3,4,8\},\{3,5,9\},\{4,5,6\}.\] We fix a basis $B=\{7,8,9\}.$

We wish to send a message $\bm{m} \in A^6$, but an intruder has knowledge of the two last digits of the message (that is, $X=\{5,6\}$) , and is able to listen to $\mu$ digits of the sent message. How much information does the intruder know about the message $\bm{m}$. Obviously, the intruder knows at least $2$ digits. But the choice of the $\mu$ digits gives the intruder different amount of information.

The matroid associated to the code $D_{X,M}$ is a matroid of rank $7$ on $E$, with the following bases: \[\begin{gathered}\{ 1, 2, 3, 4, 5, 6, 8 \},
    \{ 1, 2, 3, 4, 5, 7, 9 \},
     \{ 1, 2, 3, 4, 5, 7, 8 \},\\
     \{ 1, 2, 3, 4, 6, 8, 9 \},
     \{ 1, 2, 3, 4, 5, 8, 9 \},
     \{ 1, 2, 3, 4, 5, 6, 7 \},
     \{ 1, 2, 3, 4, 6, 7, 8 \},\\
     \{ 1, 2, 3, 4, 5, 6, 9 \},
     \{ 1, 2, 3, 4, 6, 7, 9 \},
     \{ 1, 2, 3, 4, 7, 8, 9 \}
\end{gathered}\]

The associated demi-matroid of the system has the following sets of profiles:

\[\begin{array}{ccccc} \sigma_0=0,&\sigma_1=3, &\sigma_2=5,&\sigma_3=6, &\sigma_4=7\end{array}.\]

From Theorem~\ref{summingup}, we know that no matter the choice of the digits an intruder listens to, if he listens to $0$, $1$, or $2$ digits, the intruder gets no information whatsoever on the sent message (except on the digits he already knows). If he is able to listen at most $4$ digits, he gets at most $1$ digit of extra information. For example, if the intruder listens to digits $4,5,6$ of the sent message $\bm{m}$, he gets one extra digit of information (in this case, the intruder actually knows the $4$th digit). If the intruder listens to the digits $3,4,8$, then he also gets one extra digit of information (in the sense of Remark~\ref{digit}, not an actual extra digit). While if he listens to digits $1,2,9$, he doesn't get any more information.


\begin{thebibliography}{1}

\bibitem{BJM} Britz, T., Johnsen, T., Martin, J.A.:  Chains, Demi-matroids and Profiles. IEEE Trans. Inform. Theory, 60, no. 2, 986--991 (2014)

\bibitem{BJMS} Britz, T., Johnsen, T., Mayhew D., Shiromoto K.: Wei-type duality theorems for matroids. Des. Codes Cryptogr., 62, no. 3, 331--341 (2012)
 

\bibitem{gordon12} Gordon, G.: On Brylawski's Generalized Duality. Math. Comput. Sci.,  6, no. 2,  135--146 (2012)

\bibitem{JV} Johnsen, T., Verdure, H.: Generalized Hamming weights for almost affine codes. IEEE Trans. Inform. Theory, 63, no. 4, 1941--1953 (2017)

\bibitem{LCL} Liu, Z., Chen, W., Luo, Y.: The relative generalized Hamming weight of linear $q$-ary codes and their subcodes, Des. Codes Cryptogr., 48, no. 2, 111-123 (2008).

\bibitem{LMVC} Luo, Y., Mitrpant, C., Han Vinck, A.J., Chen, K.: Some New Characters on the Wire-Tap Channel of Type II. IEEE Trans. Inform. Theory,  51, no. 3 , 1222--1227 (2005)

\bibitem{O} Oxley, J.G.: {Matroid theory}. Oxford university press, New York (2011)

\bibitem{OW} Ozarow, L.H., Wyner, A.D.: Wire-Tap Channel II. Advances in cryptology (Paris, 1984), 33--50, Lecture Notes in Comput. Sci., 209, Springer, Berlin, (1985)

\bibitem {SA} Simonis, J., Ashikhmin, A.: Almost Affine Codes. Des. Codes Cryptogr., 14, no. 2, 179--197 (1998)

\bibitem{wei91} Wei, V.K.: Generalized Hamming weights for linear codes. IEEE Trans. Inform. Theory,  37, no. 5,  1412--1418 (1991)

\bibitem{ZBLH} Zhuang, Z., Dai, B., Luo, Y., Han Vinck, A.J.: On the relative profiles of a linear code and a subcode. Des. Codes Cryptogr., 72, no. 2,  219--247 (2014)

\end{thebibliography}
\end{document}